\documentclass[
  aip,
  jmp,
  amsmath,amssymb,
  onecolumn
]{revtex4-2}

\usepackage{amsmath,amssymb,amsthm,mathtools}
\usepackage{booktabs}
\usepackage{graphicx}
\usepackage{microtype}
\usepackage[hidelinks]{hyperref}

\theoremstyle{plain}
\newtheorem{theorem}{Theorem}[section]
\newtheorem{proposition}[theorem]{Proposition}
\newtheorem{corollary}[theorem]{Corollary}
\newtheorem{lemma}[theorem]{Lemma}

\theoremstyle{definition}
\newtheorem{assumption}[theorem]{Assumption}
\newtheorem{remark}[theorem]{Remark}

\newcommand{\ii}{\mathrm{i}}
\newcommand{\id}{\mathbf{1}}
\newcommand{\cB}{\mathcal{B}}
\newcommand{\cR}{\mathcal{R}}
\newcommand{\cM}{\mathcal{M}}
\newcommand{\cH}{\mathcal{H}}
\newcommand{\tr}{\operatorname{tr}}
\newcommand{\diag}{\operatorname{diag}}

\begin{document}

\title{Jordan-Twist Bethe Ansatz and Many-Body
Exceptional-Point Amplification in the Finite-\texorpdfstring{$U$}{U}
Anderson Impurity}

\author{Vinayak M.~Kulkarni}
\affiliation{Theoretical Sciences Unit,
Jawaharlal Nehru Centre for Advanced Scientific Research,
Jakkur, Bangalore 560064, India}
\email{vmkphysimath@gmail.com}

\begin{abstract}
We construct an interacting integrable impurity whose local
second-order exceptional point is amplified into higher-order
many-body Jordan blocks.  The construction starts from a pair of
spin--orbit branches, linearizes about counterpropagating Fermi
points, and folds them into two equal-velocity chiral components.
The impurity carries the constant pseudo-Hermitian matrix
$M=\gamma\sigma_x+\ii\beta\sigma_z$, the hybridization is a component
scalar, and the local interaction is the standard finite-$U$
Anderson term.  A canonical $GL(2,\mathbb C)$ transformation leaves
that interaction invariant.  A position-dependent version maps the
model to the conventional Anderson Hamiltonian plus a conserved
global $\mathfrak{gl}_2$ charge and induces the boundary twist
$G=\exp(\ii ML/v)$.

The bulk two-electron matrix is therefore the exact finite-$U$
Anderson $R$-matrix.  It has rational difference form in the dressed
rapidity
$u(p)=p(p-2\epsilon_d-U)/(2U\Gamma_A)$, satisfies the Yang--Baxter
equation, and generates RLL and twisted RTT relations for arbitrary
particle number.  At $\beta^2=\gamma^2$ the twist is nontrivial
unipotent.  We prove that, on every copy of a global pseudospin-$S$
multiplet with Anderson energy $E_{\alpha S}(U)$, the deformed
Hamiltonian is similar to the single Jordan block
$J_{2S+1}(E_{\alpha S})$.  Thus a local $2\times2$ exceptional point
produces an exact many-body exceptional point of order $2S+1$; finite
$U$ changes the multiplet energy but cannot split this block.

We also connect the Jordan chain to the finite-$U$ nested Bethe
ansatz.  Upon approaching the unipotent twist through a singularly
conjugated diagonal twist, the $r$ descendant spin roots scale as
$\lambda_a=x_a/\phi+O(1)$.  Their leading positions satisfy a
Stieltjes system whose polynomial is
$L_r^{(-2S-1)}(2x)$.  All descendant roots therefore meet at the
single projective point $\lambda=\infty$, whereas the finite charge
and highest-weight spin roots retain their full $U$ dependence.
For the maximal descendant $r=2S$, a further large-multiplet limit
maps this Laguerre polynomial exactly to a truncated exponential:
the normalized coordinates $z_a=2x_a/r$ acquire a continuum
zero distribution on the Szeg\H{o} curve
$|z e^{1-z}|=1$, $|z|\leq1$.
The construction is exact for the static equal-velocity linearized
model; curvature, unequal velocities, and nonscalar channel
couplings lie outside the theorem.
\end{abstract}

\maketitle

\section{Introduction}
\label{sec:introduction}

Yang--Baxter integrability and non-Hermitian diagonalizability are
independent algebraic properties.  Exceptional points and
pseudo-Hermitian operators provide a natural setting in which this
distinction becomes spectral rather than merely formal
~\cite{Mostafazadeh2002,Rotter2009,Heiss2012,Ashida2020}.  The first
property is controlled by factorized scattering or, equivalently, an
RLL/RTT algebra; the second concerns the spectral representation of
the resulting commuting operators.  An integrable transfer matrix
may therefore be defective and require generalized eigenvectors
rather than an ordinary Bethe eigenbasis
~\cite{Bethe1931,Yang1967,Baxter1982,Faddeev1996,Korepin1993,
NietoGarcia2024,BorsatoGarcia2025}.

The finite-$U$ Anderson impurity is a useful setting in which to make
this distinction precise.  Its exact two-electron matrix is not of
difference form in the bare momenta.  It becomes the rational
$\mathfrak{gl}_2$ matrix only after introducing the nonlinear
Anderson rapidity~\cite{Wiegmann1980,KawakamiOkiji1982,
WiegmannTsvelick1983Exact,TsvelickWiegmann1983,ChaoPalacios2011}.
Consequently, neither a bare pair-energy
denominator nor a scalar Breit--Wigner phase is by itself a valid
test of interacting factorizability.

The second ingredient is a constant pseudo-Hermitian impurity matrix
\begin{equation}
 M=\gamma\sigma_x+\ii\beta\sigma_z,\qquad
 M^2=(\gamma^2-\beta^2)\id .
 \label{eq:intro-M}
\end{equation}
Away from $\beta^2=\gamma^2$, this matrix is diagonalizable by a
generally non-unitary similarity transformation.  At
$\beta^2=\gamma^2$ it is nonzero and nilpotent, so the diagonalizing
transformation is singular~\cite{Mostafazadeh2002,Heiss2012,Kato1995}.
A proof performed only in the diagonal basis is therefore
insufficient at the exceptional point (EP).

Non-Hermitian Kondo and Anderson problems have recently been studied
by Bethe-ansatz, renormalization-group, and numerical methods
~\cite{Nakagawa2018,Kulkarni2022,Yamamoto2025,Kattel2025,
Pasnoori2025,BurkeMitchell2025}.  The spin--orbit and driven-Dirac
motivation for the constant matrix in Eq.~\eqref{eq:intro-M} is
described in Ref.~\cite{KulkarniDirac2026}.  The present result is
deliberately narrower: it establishes an exact algebraic construction
for the static equal-velocity linearization and identifies the
resulting higher-order EP as symmetry protected.

Here the singularity is handled without diagonalizing $M$.  The
linearized impurity is mapped exactly to
\begin{equation}
 H_A+Q_M,
 \label{eq:intro-deformation}
\end{equation}
where $H_A$ is the conventional finite-$U$ Anderson Hamiltonian and
$Q_M$ is one of its conserved complexified pseudospin charges.  On a
ring, the same map moves $M$ into the boundary twist
$G=\exp(\ii ML/v)$.  Since the rational Anderson $R$-matrix commutes
with $G\otimes G$ for every $G\in GL(2,\mathbb C)$, the interacting
RTT algebra survives at the unipotent EP.

The principal new result is stronger than the survival of commuting
transfer matrices.  Global $GL(2,\mathbb C)$ covariance fixes the
Jordan structure of Eq.~\eqref{eq:intro-deformation}: every
pseudospin-$S$ Anderson multiplet becomes one Jordan block of length
$2S+1$.  The finite-$U$ nested equations give the complementary
Bethe description.  Descendants are represented by roots that run to
infinity as a diagonal twist is removed, a phenomenon known for
rational spin chains~\cite{KazakovLeurentVolin2016}.  Here that
contraction is embedded into the interacting Anderson rapidity map
and tied to a physical unipotent impurity twist.  Its leading
Laguerre polynomial is independent of $U$.  In the maximal-descendant
large-multiplet limit, the resulting normalized roots have an
explicit continuum density supported on the Szeg\H{o} curve.

Arbitrary twists of the XXX monodromy, roots at infinity, and
Jordanian spin-chain deformations are established topics
individually~\cite{BelliardSlavnovVallet2018,
KazakovLeurentVolin2016,NietoGarcia2024,BorsatoGarcia2025}.  The
model-specific contribution of this work is their exact combination:
(i) the gauge embedding of the pseudo-Hermitian finite-$U$ Anderson
impurity into a unipotently twisted rational monodromy,
(ii) the all-multiplet Jordan-block theorem, and
(iii) the Anderson-dressed root contraction demonstrating why the
standard local interaction does not lift the many-body EP.

\section{Linearized model and its exact symmetry}
\label{sec:model}

\subsection{Spin--orbit linearization and folding}

Consider the two branches
\begin{equation}
 \varepsilon_\eta(k)=\alpha k^2+\eta\lambda k
 =\alpha\left(k+\eta\frac{\lambda}{2\alpha}\right)^2
  -\frac{\lambda^2}{4\alpha},
 \qquad \eta=\pm .
 \label{eq:shifted-parabolas}
\end{equation}
They are momentum-shifted copies of the same parabola.  Linearizing
about counterpropagating Fermi points gives
\begin{equation}
 \varepsilon_\eta(k_{\eta F}+q)
 =\mu+\eta vq+O(q^2),\qquad v>0 .
 \label{eq:linearization}
\end{equation}
After subtracting $\mu$ and folding the left mover,
$\chi_2(x)=\psi_-(-x)$, both components have the same right-moving
kinetic term
\begin{equation}
 H_c=-\ii v\int_0^L dx\,
       \bar\chi_a(x)\partial_x\chi_a(x).
 \label{eq:folded-bath}
\end{equation}
A bar denotes the canonical dual.  Under a non-unitary canonical
transformation it transforms with the inverse matrix, not with the
Hermitian adjoint.

\begin{assumption}[Integrable reduction]
\label{ass:reduction}
The model studied below has equal folded velocities, a linear
wide-band continuum, a static component matrix $M$, and a
momentum-independent scalar hybridization.  The curvature
$O(q^2)$, unequal velocities, energy-dependent self-energies, and
additional component-dependent vertices are omitted.
\end{assumption}

Under Assumption~\ref{ass:reduction},
\begin{align}
 H(M)={}&-\ii v\int_0^L dx\,
       \bar\chi_a(x)\partial_x\chi_a(x)
 +\bar d_a(\epsilon_d\delta_{ab}+M_{ab})d_b
 +U n_{d1}n_{d2}+H_{\rm hyb},                         \label{eq:H-M}\\
 H_{\rm hyb}={}&V\bigl[\bar\chi_a(0)d_a+
                 \bar d_a\chi_a(0)\bigr],\qquad
 n_{da}=\bar d_a d_a,                                 \label{eq:H-hyb}\\
 M={}&\gamma\sigma_x+\ii\beta\sigma_z,\qquad
 \Gamma_A=\frac{V^2}{2v}.                             \label{eq:M-Gamma}
\end{align}
The last equality fixes the continuum normalization of the
hybridization width.

\subsection{Canonical \texorpdfstring{$GL(2,\mathbb C)$}{GL(2,C)}
covariance}

Let $S\in GL(2,\mathbb C)$ act on bath and impurity fields by
\begin{equation}
 \chi=S\widetilde\chi,\quad d=S\widetilde d,\qquad
 \bar\chi=\bar{\widetilde\chi}S^{-1},\quad
 \bar d=\bar{\widetilde d}S^{-1}.
 \label{eq:canonical-S}
\end{equation}
This transformation is canonical even when it is not unitary.

\begin{lemma}[Interaction invariance]
\label{lem:interaction}
The finite-$U$ vertex is invariant under
Eq.~\eqref{eq:canonical-S}.
\end{lemma}

\begin{proof}
With $N_d=\bar d_a d_a$, the two-component fermion identity
\begin{equation}
 n_{d1}n_{d2}=\frac12N_d(N_d-1)
 \label{eq:double-occupancy}
\end{equation}
holds.  Equation~\eqref{eq:canonical-S} leaves $N_d$ invariant and
hence leaves Eq.~\eqref{eq:double-occupancy} invariant.
\end{proof}

Let $H_A=H(0)$ and introduce the global generators
\begin{equation}
 {\cal J}_{ab}
 =\int_0^L dx\,\bar\psi_a(x)\psi_b(x)+\bar d_a d_b .
 \label{eq:Jab}
\end{equation}
For a constant $2\times2$ matrix $X$, write
\begin{equation}
 Q_X=\int_0^L dx\,\bar\psi_aX_{ab}\psi_b+\bar d_aX_{ab}d_b .
 \label{eq:QX}
\end{equation}
It is useful to make the transverse components explicit.  Define the
global pseudospin generators
\begin{equation}
 J^\mu=\frac12Q_{\sigma_\mu},\qquad
 J^\pm=J^x\pm\ii J^y ,
 \label{eq:global-ladders}
\end{equation}
where the same ladder operators are often denoted $S^\pm$.  For the
matrix in Eq.~\eqref{eq:M-Gamma},
\begin{equation}
 Q_M
 =2\gamma J^x+2\ii\beta J^z
 =\gamma(J^++J^-)+2\ii\beta J^z .
 \label{eq:QM-ladders}
\end{equation}
Thus the nonzero transverse ladder terms are retained in the
Hamiltonian; they are not removed by an assumption.

\begin{proposition}[Complexified pseudospin symmetry]
\label{prop:GL2}
For every $X\in\mathfrak{gl}_2(\mathbb C)$,
\begin{equation}
 [H_A,Q_X]=0.
 \label{eq:global-commutator}
\end{equation}
\end{proposition}

\begin{proof}
The bath dispersion, impurity level, and hybridization contract
component indices with the identity.  Lemma~\ref{lem:interaction}
gives the invariance of the only quartic term.  Their infinitesimal
global $GL(2,\mathbb C)$ variation is therefore zero.
\end{proof}
For non-Hermitian $X$, ``symmetry'' here means a commuting,
generally non-self-adjoint charge rather than a unitary symmetry.
Away from the EP a global $GL(2,\mathbb C)$ transformation reduces
Eq.~\eqref{eq:QM-ladders} to $2sJ^z$; at the EP it becomes a
nonzero nilpotent ladder generator.  The explicit off-EP
diagonalizer, finite-ring domain map, and EP limit are given in the
Supplementary Material.

\subsection{Gauge equivalence and the finite-volume domain}

\begin{theorem}[Gauge embedding]
\label{thm:gauge}
The canonical field transformation
\begin{equation}
 \psi(x)=e^{-\ii Mx/v}\chi(x),\qquad
 \bar\psi(x)=\bar\chi(x)e^{\ii Mx/v}
 \label{eq:gauge-transform}
\end{equation}
maps $H_A+Q_M$ to the local Hamiltonian $H(M)$.  If $\psi$ is
periodic on a ring of length $L$, the image field satisfies
\begin{equation}
 \chi(L)=G\chi(0),\qquad G=e^{\ii ML/v}.
 \label{eq:boundary-twist}
\end{equation}
\end{theorem}

\begin{proof}
The kinetic density transforms as
\begin{equation}
 -\ii v\,\bar\psi\partial_x\psi
 =-\ii v\,\bar\chi\partial_x\chi-\bar\chi M\chi .
 \label{eq:gauge-kinetic}
\end{equation}
This cancels the bath part of $Q_M$, while the impurity part becomes
$\bar dMd$.  The transformation equals the identity at $x=0$, so
the scalar hybridization is unchanged.  Lemma~\ref{lem:interaction}
handles the interaction.  Periodicity of $\psi$ gives
Eq.~\eqref{eq:boundary-twist}.
\end{proof}

\begin{remark}
The spectral theorem below concerns the periodic problem
$H_A+Q_M$, equivalently $H(M)$ with the induced boundary condition
\eqref{eq:boundary-twist}.  Imposing a different boundary condition
on $H(M)$ defines a different finite-volume spectral problem.  On
every fixed-$N_e$ sector of a finite ring, the induced
first-quantized multiplication operator is bounded, everywhere
invertible, and maps the periodic operator domain bijectively to the
$G$-twisted domain; the Supplementary Material gives the explicit
operator and norm bounds.  Consequently, none of the finite-volume
spectral or Jordan statements below relies on an unbounded
infinite-line transformation.  On the infinite line the map is only
an algebraic scattering equivalence once its asymptotic domain has
been fixed.
\end{remark}

\section{Frozen quadratic benchmark}
\label{sec:frozen}

At $U=0$, the wide-band one-particle impurity matrix is
\begin{equation}
 K_0(\omega)=\id-2\ii\Gamma_A
 \bigl[(x+\ii\Gamma_A)\id-M\bigr]^{-1},
 \qquad x=\omega-\epsilon_d .
 \label{eq:K0}
\end{equation}
Writing
\begin{equation}
 s^2=\gamma^2-\beta^2,\qquad M^2=s^2\id ,
 \label{eq:M-square}
\end{equation}
gives
\begin{equation}
 \bigl[(x+\ii\Gamma_A)\id-M\bigr]^{-1}
 =\frac{(x+\ii\Gamma_A)\id+M}
 {(x+\ii\Gamma_A)^2-s^2}.
 \label{eq:resolvent}
\end{equation}
Thus $K_0(\omega)=A(\omega)\id+B(\omega)M$ and
\begin{equation}
 [K_0(\omega),K_0(\omega')]=0.
 \label{eq:K0-commute}
\end{equation}
At $s=0$ the matrix is Jordan but remains a rational function of the
same nilpotent $M$.

This frozen limit is an exact consistency check, not the interacting
integrability proof.  Its $N$-fermion scattering is the
antisymmetrized product of one-particle matrices and is therefore
Gaussian.  The finite-$U$ problem requires the nontrivial exchange
matrix constructed next.

\section{Finite-\texorpdfstring{$U$}{U} Yang--Baxter and RTT
construction}
\label{sec:finiteU}

\subsection{Dressed Anderson rapidity}

With momentum measured in energy units, the exact finite-$U$
two-electron component matrix is
\begin{equation}
 R_{12}(p,q)=
 \frac{[\cB(p)-\cB(q)]\id_{12}
       +\ii\,2U\Gamma_A P_{12}}
      {\cB(p)-\cB(q)+\ii\,2U\Gamma_A},
 \qquad
 \cB(p)=p(p-2\epsilon_d-U),
 \label{eq:Anderson-R}
\end{equation}
where $P_{12}$ exchanges the two pseudospins.  The variable
$\cB(p)$ is the standard exact-Anderson rapidity
~\cite{KawakamiOkiji1982,WiegmannTsvelick1983Exact}; the displayed
normalized two-particle matrix is written explicitly in
Ref.~\cite{ChaoPalacios2011}.  The triplet eigenvalue is one and the
singlet eigenvalue is
\begin{equation}
 R_{\rm s}(p,q)
 =\frac{\cB(p)-\cB(q)-\ii\,2U\Gamma_A}
        {\cB(p)-\cB(q)+\ii\,2U\Gamma_A}.
 \label{eq:singlet-R}
\end{equation}

For $U\ne0$, define
\begin{equation}
 u(p)=\frac{\cB(p)}{2U\Gamma_A}.
 \label{eq:dressed-rapidity}
\end{equation}
Then, up to scalar normalization,
\begin{equation}
 \cR_{12}(u)=u\,\id_{12}+\ii P_{12},\qquad
 u=u(p)-u(q).
 \label{eq:rational-R}
\end{equation}
The identity
\begin{equation}
 \cB(p)-\cB(q)=(p-q)(p+q-2\epsilon_d-U)
 \label{eq:B-factor}
\end{equation}
shows why a pair-energy factor in bare momenta does not obstruct
factorization: the exact difference variable is $u(p)-u(q)$.

The permutation algebra gives
\begin{align}
 &\cR_{12}(u-v)\cR_{13}(u-w)\cR_{23}(v-w)\notag\\
 &\qquad =
 \cR_{23}(v-w)\cR_{13}(u-w)\cR_{12}(u-v).
 \label{eq:YBE}
\end{align}
Moreover, for every $S\in GL(2,\mathbb C)$,
\begin{equation}
 (S\otimes S)^{-1}\cR_{12}(u)(S\otimes S)=\cR_{12}(u),
 \label{eq:R-GL2}
\end{equation}
because $[P_{12},S\otimes S]=0$.  A CPT or metric rotation therefore
preserves the ordinary Yang--Baxter equation.  In this rational
model a symmetric $S\otimes S$ rotation does not produce a new bulk
$R$-matrix; the non-Hermitian datum is carried by the boundary twist.

\subsection{RLL, RTT, and arbitrary particle number}

For electron $j$, set
\begin{equation}
 \xi_j=u(p_j),\qquad
 L_{aj}(z)=\cR_{aj}(z-\xi_j),\qquad
 \cM_a(z)=L_{aN_e}(z)\cdots L_{a1}(z).
 \label{eq:Lax-monodromy}
\end{equation}
The Yang--Baxter equation implies
\begin{equation}
 \cR_{ab}(z-w)L_{aj}(z)L_{bj}(w)
 =L_{bj}(w)L_{aj}(z)\cR_{ab}(z-w).
 \label{eq:RLL}
\end{equation}
Define
\begin{equation}
 T_a(z)=G_a\cM_a(z),\qquad
 t_G(z)=\tr_aT_a(z).
 \label{eq:twisted-transfer}
\end{equation}

\begin{theorem}[Finite-$U$ twisted RTT algebra]
\label{thm:RTT}
For every constant invertible $G\in GL(2,\mathbb C)$,
diagonalizable or not,
\begin{equation}
 \cR_{ab}(z-w)T_a(z)T_b(w)
 =T_b(w)T_a(z)\cR_{ab}(z-w),
 \label{eq:RTT}
\end{equation}
and hence
\begin{equation}
 [t_G(z),t_G(w)]=0
 \label{eq:transfer-commute}
\end{equation}
for arbitrary $N_e$.
\end{theorem}

\begin{proof}
Equation~\eqref{eq:R-GL2} gives
$[\cR_{ab}(z-w),G_aG_b]=0$.  Multiplying
Eq.~\eqref{eq:RLL} over all physical spaces gives the untwisted RTT
relation.  Inserting $G_aG_b$, commuting it through $\cR_{ab}$, and
taking the two auxiliary traces proves the result.
\end{proof}

The charge sector enters through the nonlinear inhomogeneities
$\xi_j$ and scalar scattering factors; the matrix RTT algebra is the
nested pseudospin sector.  This continuum
construction does not identify the four-state Anderson impurity with
a homogeneous XXX lattice site, nor does it require a Shastry-type
four-state Lax operator.

There is nevertheless a precise four-state extension of the
\emph{twist}, distinct from such an identification.  On the local
Fock space
$\mathcal F_d=\{|0\rangle,|\uparrow\rangle,|\downarrow\rangle,
|\uparrow\downarrow\rangle\}$, the canonical lift of
$G\in GL(2,\mathbb C)$ is
\begin{equation}
 \widehat G=1\oplus G\oplus\det G .
 \label{eq:Fock-lift-summary}
\end{equation}
As shown explicitly in the Supplementary Material,
$\widehat G$ belongs to the symmetry family of the
$16\times16$ fermionic Hubbard $R$-matrix and hence preserves its
graded YBE and twisted RTT algebra, including when $G$ is unipotent
at the EP~\cite{Shastry1986,Shiroishi1998}.  This is an exact
Hubbard-$R$ compatibility statement, not a derivation of
Eq.~\eqref{eq:H-M} from a four-state Hubbard transfer matrix.  The
latter would additionally require an impurity $L$-operator satisfying
the graded RLL relation, or a boundary $K$-matrix satisfying the
reflection equation, whose logarithmic derivative reproduces the
specific Anderson impurity Hamiltonian without extra correlated
vertices.  Existing Hubbard-chain Anderson embeddings illustrate
that this Hamiltonian-matching condition is nontrivial~\cite{Song2025}.

\section{Exact many-body exceptional-point theorem}
\label{sec:Jordan-theorem}

\subsection{Local EP and unipotent twist}

The matrix $M$ is $\sigma_x$-pseudo-Hermitian:
\begin{equation}
 \sigma_xM^\dagger\sigma_x=M.
 \label{eq:pseudo-Hermitian}
\end{equation}
For $s\ne0$ it is similar to $s\sigma_z$.  At
\begin{equation}
 s^2=\gamma^2-\beta^2=0,\qquad M_{\rm EP}\ne0,
 \label{eq:EP-condition}
\end{equation}
it is rank-one nilpotent and is $GL(2,\mathbb C)$-conjugate to
$\nu E_{12}$ for some nonzero normalization $\nu$:
\begin{equation}
 C^{-1}M_{\rm EP}C=\nu E_{12},\qquad E_{12}^2=0.
 \label{eq:nilpotent-normal}
\end{equation}
The boundary matrix remains invertible,
\begin{equation}
 G_{\rm EP}=e^{\ii M_{\rm EP}L/v}
 =\id+\ii\frac{L}{v}M_{\rm EP},\qquad
 G_{\rm EP}^{-1}=\id-\ii\frac{L}{v}M_{\rm EP}.
 \label{eq:GEP}
\end{equation}
It is a nontrivial unipotent Jordan matrix.

\subsection{Multiplet decomposition}

Fix total electron number $N_e$.  The global pseudospin action
decomposes the Hilbert space into multiplicity and irreducible
factors,
\begin{equation}
 \cH_{N_e}
 =\bigoplus_S {\cal M}_S\otimes V_S,\qquad
 \dim V_S=2S+1 .
 \label{eq:multiplet-decomposition}
\end{equation}
For a continuum, the direct sum over orbital labels can be understood
as the corresponding spectral direct integral.  Since $H_A$ commutes
with the full pseudospin algebra, it has the form
\begin{equation}
 H_A\big|_{{\cal M}_S\otimes V_S}
 =h_S(U)\otimes\id_{V_S}.
 \label{eq:Schur-form}
\end{equation}
Let $\alpha$ label an eigenvector of $h_S(U)$ with eigenvalue
$E_{\alpha S}(U)$.  The associated invariant copy of $V_S$ will be
denoted ${\cal E}_{\alpha S}$.

\begin{theorem}[Many-body EP amplification]
\label{thm:manybody-EP}
Let $M_{\rm EP}$ satisfy Eq.~\eqref{eq:EP-condition}.  On every
Anderson multiplet ${\cal E}_{\alpha S}$,
\begin{equation}
 (H_A+Q_{M_{\rm EP}})\big|_{{\cal E}_{\alpha S}}
 \ \sim\
 E_{\alpha S}(U)\id_{2S+1}+\nu J^+_S
 \ \sim\
 J_{2S+1}\!\left(E_{\alpha S}(U)\right).
 \label{eq:main-Jordan}
\end{equation}
Thus the local second-order EP is amplified into exactly one
many-body Jordan block of order $2S+1$ on each pseudospin-$S$
multiplet.
\end{theorem}

\begin{proof}
Apply the global canonical transformation induced by $C$ in
Eq.~\eqref{eq:nilpotent-normal}.  Proposition~\ref{prop:GL2} leaves
$H_A$ unchanged, while $Q_{M_{\rm EP}}$ becomes $\nu J^+$.  By
Eq.~\eqref{eq:Schur-form}, $H_A$ is
$E_{\alpha S}(U)\id$ on ${\cal E}_{\alpha S}$.  In the standard
weight basis,
\begin{equation}
 J^+_S|S,m\rangle
 =\sqrt{(S-m)(S+m+1)}\,|S,m+1\rangle .
 \label{eq:Jplus-action}
\end{equation}
Every superdiagonal coefficient is nonzero for
$m=-S,\ldots,S-1$.  Hence $J^+_S$ has nilpotency index $2S+1$ and
geometric multiplicity one.  It is similar to one nilpotent Jordan
block of that size, proving Eq.~\eqref{eq:main-Jordan}.
\end{proof}

\begin{corollary}[Explicit Jordan chain]
\label{cor:Jordan-chain}
With $r=0,\ldots,2S$, define
\begin{equation}
 |v_r\rangle
 =\frac{|S,S-r\rangle}
 {\nu^r\sqrt{r!\,(2S)!/(2S-r)!}} .
 \label{eq:explicit-chain}
\end{equation}
Then
\begin{equation}
 (H_A+Q_{M_{\rm EP}}-E_{\alpha S})|v_0\rangle=0,\qquad
 (H_A+Q_{M_{\rm EP}}-E_{\alpha S})|v_r\rangle
 =|v_{r-1}\rangle .
 \label{eq:chain-relation}
\end{equation}
\end{corollary}

\begin{corollary}[Finite-$U$ robustness]
\label{cor:U-robustness}
The standard Anderson interaction changes
$E_{\alpha S}(U)$ and the orbital part of the state, but it cannot
reduce or split the internal Jordan block in
Eq.~\eqref{eq:main-Jordan}.  The largest allowed block at fixed
$N_e$ has size $N_e+1$, whenever an $S=N_e/2$ multiplet occurs.
\end{corollary}

Away from the EP, choose $C_s$ with
$C_s^{-1}M C_s=s\sigma_z$.  The same multiplet has the
diagonalizable spectrum
\begin{equation}
 E_{\alpha S,m}(s)=E_{\alpha S}(U)+2sm,
 \qquad m=-S,\ldots,S.
 \label{eq:multiplet-splitting}
\end{equation}
All $2S+1$ eigenvalues and eigenvectors coalesce into
Eq.~\eqref{eq:main-Jordan} as the diagonalizing transformation becomes
singular.  The unfolding \eqref{eq:multiplet-splitting} is constrained
by symmetry; it is not the generic one-parameter unfolding of an
arbitrary $(2S+1)$-dimensional Jordan block.

\section{Singular twist and the commuting Jordan family}
\label{sec:singular-twist}

\subsection{A regular family with a singular diagonalizer}

The local normal form of the boundary degeneration can be written
\begin{align}
 D_\phi&=\diag(e^{\ii\phi},e^{-\ii\phi}),\notag\\
 C_\phi&=
 \begin{pmatrix}
 1&\displaystyle\frac{\rho}{e^{-\ii\phi}-e^{\ii\phi}}\\[4pt]
 0&1
 \end{pmatrix},                                        \label{eq:Cphi}\\
 G_\phi&=C_\phi D_\phi C_\phi^{-1}
 =\begin{pmatrix}e^{\ii\phi}&\rho\\0&e^{-\ii\phi}\end{pmatrix}
 \longrightarrow
 G_0=\id+\rho E_{12}.                                  \label{eq:Gphi}
\end{align}
Here $\rho\ne0$ is fixed.  The matrices $G_\phi$ have a finite
unipotent limit, while $C_\phi=O(\phi^{-1})$ is singular.  The
physical twist \eqref{eq:GEP} is conjugate to $G_0$, with
$\phi=sL/v$ away from the EP.

Let
\begin{equation}
 {\cal C}_\phi=C_\phi^{\otimes N_e}.
 \label{eq:global-Cphi}
\end{equation}
The $GL(2)$ covariance of the rational Lax operator gives
\begin{equation}
 t_{G_\phi}(z)
 ={\cal C}_\phi\,t_{D_\phi}(z)\,
  {\cal C}_\phi^{-1}.
 \label{eq:transfer-similarity}
\end{equation}
Indeed,
$C_{\phi,a}^{-1}L_{aj}C_{\phi,a}
=C_{\phi,j}L_{aj}C_{\phi,j}^{-1}$ follows directly from
Eq.~\eqref{eq:R-GL2}.  Multiplication over $j$ and cyclicity of the
auxiliary trace then give Eq.~\eqref{eq:transfer-similarity}.
For $\phi\ne0$ this is an ordinary similarity.  At $\phi=0$ it is a
singular confluent limit, which preserves eigenvalues but can merge
eigenvectors.

\subsection{Exact isospectrality and a joint Jordan chain}

Write the untwisted monodromy as
\begin{equation}
 \cM_a(z)=
 \begin{pmatrix}A(z)&B(z)\\C(z)&D(z)\end{pmatrix}_a .
 \label{eq:ABCD}
\end{equation}
For the convention in Eq.~\eqref{eq:Gphi},
\begin{equation}
 t_{G_0}(z)=A(z)+D(z)+\rho C(z).
 \label{eq:unipotent-transfer}
\end{equation}
The operator $C(z)$ changes total pseudospin weight by one.  Ordering
the quantum space by weight makes Eq.~\eqref{eq:unipotent-transfer}
block triangular, with the same diagonal weight blocks as the
periodic transfer matrix $t_{\id}(z)=A(z)+D(z)$.

\begin{proposition}[Unipotent isospectrality]
\label{prop:isospectral}
For every $z$, every set of inhomogeneities, and every $\rho$,
\begin{equation}
 \det\!\left[\tau\id-t_{G_0}(z)\right]
 =\det\!\left[\tau\id-t_{\id}(z)\right].
 \label{eq:charpoly-equality}
\end{equation}
The unipotent twist changes the Jordan structure but not the
characteristic polynomial.
\end{proposition}

\begin{proof}
In a weight-ordered basis the added term $\rho C(z)$ has no diagonal
weight block.  The determinant of the resulting block-triangular
matrix is the product of the determinants of the unchanged diagonal
blocks.
\end{proof}

The large-$z$ expansion makes the nilpotent generator explicit:
\begin{align}
 t_{G_0}(z)
 ={}&2z^{N_e}
 +z^{N_e-1}
 \left[
  -2\sum_{j=1}^{N_e}\xi_j+\ii N_e
 \right]\id\notag\\
 &+\ii\rho\,z^{N_e-1}J^+
 +O(z^{N_e-2}).
 \label{eq:large-z-transfer}
\end{align}
Since all coefficients of $t_{G_0}(z)$ commute, the commuting algebra
itself contains $J^+$.  On every $V_S$, this coefficient supplies the
exact length-$(2S+1)$ joint Jordan chain of
Corollary~\ref{cor:Jordan-chain}.

\begin{proposition}[Generic transfer-matrix amplification]
\label{prop:generic-transfer}
Fix $z$ and suppose an eigenvalue $\tau_{\alpha S}(z)$ of
$t_{\id}(z)$ is isolated from all periodic multiplets other than one
copy of $V_S$.  Let $P_{\alpha S}$ be its spectral projector.  If
\begin{equation}
 P_{\alpha S}
 \left.\frac{\partial t_{D_\phi}(z)}{\partial\phi}
 \right|_{\phi=0}
 P_{\alpha S}
 =\kappa_{\alpha S}(z)J^z_S,\qquad
 \kappa_{\alpha S}(z)\ne0,
 \label{eq:twist-derivative}
\end{equation}
then the spectral subspace of $t_{G_0}(z)$ associated with
$\tau_{\alpha S}(z)$ contains one Jordan block of size $2S+1$.
\end{proposition}

\begin{proof}
The derivative in Eq.~\eqref{eq:twist-derivative} is the zero
component of a global pseudospin vector.  Its projection onto a
single irreducible $V_S$ is therefore proportional to $J^z_S$.
On this factor,
${\cal C}_\phi=\exp(a_\phi J^+_S)$ with
\begin{equation}
 a_\phi=\frac{\rho}{e^{-\ii\phi}-e^{\ii\phi}}
 =\frac{\ii\rho}{2\phi}+O(\phi).
 \label{eq:a-phi}
\end{equation}
Using
$e^{aJ^+}J^ze^{-aJ^+}=J^z-aJ^+$ in
Eq.~\eqref{eq:transfer-similarity}, the finite confluent limit on the
spectral subspace has the form
\begin{equation}
 t_{G_0}(z)
 =\tau_{\alpha S}(z)\id+
 f_{\alpha S}(J^+_S),\qquad
 f_{\alpha S}'(0)=-\frac{\ii\rho}{2}
 \kappa_{\alpha S}(z)\ne0.
 \label{eq:polynomial-Jplus}
\end{equation}
Higher orders in $\phi$ can add higher powers of $J^+_S$ but cannot
cancel the displayed linear term.  A polynomial $f(N)$ of one
nilpotent Jordan block $N$, with $f(0)=0$ and $f'(0)\ne0$, has the
same Jordan partition as $N$.  The block therefore has size $2S+1$.
\end{proof}

The coefficient $\kappa_{\alpha S}(z)$ is not identically zero:
the same large-$z$ expansion gives
$\kappa_{\alpha S}(z)=-2z^{N_e-1}+O(z^{N_e-2})$.
Proposition~\ref{prop:generic-transfer} consequently
applies for generic inhomogeneities and all but exceptional spectral
values $z$.  At such exceptional values, or at accidental multiplet
degeneracies, an individual transfer matrix can have a different
block partition even though the common commuting algebra and the
Hamiltonian chain remain fixed.  This distinction is necessary:
commuting non-diagonalizable transfer matrices and their Hamiltonians
need not have identical Jordan decompositions in every integrable
model~\cite{NietoGarcia2024}.

\section{Finite-\texorpdfstring{$U$}{U} Bethe roots at the Jordan
boundary}
\label{sec:Bethe}

\subsection{Twisted nested equations}

For $\phi\ne0$, $G_\phi$ is similar to the diagonal twist $D_\phi$,
so its spectrum is described by the standard twisted nested Bethe
ansatz~\cite{WiegmannTsvelick1983Exact,ChaoPalacios2011,
BelliardSlavnovVallet2018}.  One
convenient normalization, with the scalar one-body phase absorbed
into the charge quantum numbers, is
\begin{equation}
 e^{\ii p_jL/v}
 =e^{\ii\phi}\prod_{\alpha=1}^{M}
 \frac{u(p_j)-\lambda_\alpha+\ii/2}
      {u(p_j)-\lambda_\alpha-\ii/2},
 \qquad j=1,\ldots,N_e ,
 \label{eq:charge-BAE}
\end{equation}
and
\begin{equation}
 e^{2\ii\phi}
 \prod_{j=1}^{N_e}
 \frac{\lambda_\alpha-\xi_j+\ii/2}
      {\lambda_\alpha-\xi_j-\ii/2}
 =
 \prod_{\substack{\beta=1\\\beta\ne\alpha}}^{M}
 \frac{\lambda_\alpha-\lambda_\beta+\ii}
      {\lambda_\alpha-\lambda_\beta-\ii},
 \qquad \xi_j=u(p_j).
 \label{eq:spin-BAE}
\end{equation}
Changing the orientation of the twist or the scattering convention
inverts both sides and reverses the signs of the rapidities without
changing any conclusion below.

At $\phi=0$, a regular highest-weight solution with $M_0$ finite spin
roots has
\begin{equation}
 S=\frac{N_e}{2}-M_0 .
 \label{eq:S-M0}
\end{equation}
Its $r$th descendant, $0\le r\le2S$, is represented for
$\phi\ne0$ by $M=M_0+r$ roots.  The additional $r$ roots diverge as
the twist is removed.

\begin{theorem}[Universal descendant-root contraction]
\label{thm:root-contraction}
Assume the highest-weight solution is regular and its $M_0$ finite
spin roots and $N_e$ charge roots have finite limits as
$\phi\to0$.  For the descendant at level $r$, the additional roots
obey
\begin{equation}
 \lambda_a=\frac{x_a}{\phi}+O(1),\qquad a=1,\ldots,r,
 \label{eq:large-root-scaling}
\end{equation}
where
\begin{equation}
 1+\frac{S}{x_a}
 -\sum_{\substack{b=1\\b\ne a}}^r
  \frac{1}{x_a-x_b}=0.
 \label{eq:Stieltjes}
\end{equation}
Equivalently, $2x_a$ are the zeros of the generalized Laguerre
polynomial
\begin{equation}
 L_r^{(-2S-1)}(2x).
 \label{eq:Laguerre}
\end{equation}
The leading contraction is independent of
$U,\epsilon_d,\Gamma_A$, and the finite inhomogeneities $\xi_j$.
\end{theorem}

\begin{proof}
Separate the $M_0$ finite roots from the $r$ roots in
Eq.~\eqref{eq:spin-BAE} and insert
$\lambda_a=x_a/\phi+O(1)$.  To first order,
\begin{align}
 e^{2\ii\phi}
 \prod_{j=1}^{N_e}
 \frac{\lambda_a-\xi_j+\ii/2}
      {\lambda_a-\xi_j-\ii/2}
 &=
 1+\ii\phi\left(2+\frac{N_e}{x_a}\right)+O(\phi^2),
 \label{eq:LHS-expansion}\\
 \prod_{\beta\ne a}
 \frac{\lambda_a-\lambda_\beta+\ii}
      {\lambda_a-\lambda_\beta-\ii}
 &=
 1+2\ii\phi\left[
 \frac{M_0}{x_a}
 +\sum_{b\ne a}\frac{1}{x_a-x_b}
 \right]+O(\phi^2).
 \label{eq:RHS-expansion}
\end{align}
Equating the coefficients and using
$N_e-2M_0=2S$ gives Eq.~\eqref{eq:Stieltjes}.

Let $y_a=2x_a$ and
$Q_r(y)=\prod_{a=1}^r(y-y_a)$.  At a simple zero,
\begin{equation}
 \frac{Q_r''(y_a)}{Q_r'(y_a)}
 =2\sum_{b\ne a}\frac{1}{y_a-y_b}.
 \label{eq:root-log-derivative}
\end{equation}
Equation~\eqref{eq:Stieltjes} is then the zero condition for
\begin{equation}
 yQ_r''+(-2S-y)Q_r'+rQ_r=0,
 \label{eq:Laguerre-ODE}
\end{equation}
whose polynomial solution is proportional to
$L_r^{(-2S-1)}(y)$~\cite{Szego1975}.
\end{proof}

\begin{corollary}[Szeg\H{o} distribution of the maximal descendant]
\label{cor:Szego}
Take the limits sequentially: first $\phi\to0$ at fixed
$r=2S$, and then $r\to\infty$.  For the maximal descendant define
\begin{equation}
 z_a=\frac{2x_a}{r},\qquad
 \mu_r=\frac{1}{r}\sum_{a=1}^{r}\delta_{z_a}.
 \label{eq:Szego-scaling}
\end{equation}
The polynomial in Theorem~\ref{thm:root-contraction} becomes
\begin{equation}
 (-1)^rL_r^{(-r-1)}(rz)
 =p_r(rz)
 :=\sum_{k=0}^{r}\frac{(rz)^k}{k!}.
 \label{eq:truncated-exponential}
\end{equation}
Consequently, $\mu_r$ converges weakly to the Szeg\H{o} measure
\begin{equation}
 d\mu_{\rm Sz}(z)
 =\frac{1}{2\pi\ii}\frac{1-z}{z}\,dz
 \quad\hbox{on}\quad
 \Gamma_{\rm Sz}
 =\left\{z\in\mathbb C:
 |z e^{1-z}|=1,\ |z|\leq1\right\},
 \label{eq:Szego-measure}
\end{equation}
with $\Gamma_{\rm Sz}$ oriented counterclockwise.
\end{corollary}

\begin{proof}
The generalized Laguerre expansion gives
\begin{equation}
 L_r^{(-r-1)}(y)
 =\sum_{k=0}^{r}
 \binom{-1}{r-k}\frac{(-y)^k}{k!}
 =(-1)^r\sum_{k=0}^{r}\frac{y^k}{k!}.
\end{equation}
Setting $y=rz$ proves Eq.~\eqref{eq:truncated-exponential}.
The weak zero distribution of these rescaled exponential sections is
the classical Szeg\H{o} theorem~\cite{Szego1924,DiazMendoza2011}.
\end{proof}

The Stieltjes--Laguerre pattern is the universal rational-spin-chain
untwisting asymptotic~\cite{KazakovLeurentVolin2016}.  Its role here
is specific: the inhomogeneities are the exact interacting Anderson
variables $\xi_j=u(p_j)$, so the expansion proves directly that
finite $U$ does not enter the leading descendant contraction.

Corollary~\ref{cor:Szego} supplies a genuine continuum root measure,
structurally analogous to the root-density limits of conventional
Bethe ansatz.  Its interpretation is nevertheless different: it
describes symmetry descendants coalescing at the projective point
$\lambda=\infty$, not finite-energy Bethe strings or merging Anderson
charge roots.  No uniform statement about a joint
$\phi\to0$, $r\to\infty$ limit is required.

\paragraph{Comparison with a finite-energy root transition.}
Related arc-to-loop rearrangements of complex Bethe roots occur in
the exactly solvable Bose--Hubbard chain with unidirectional hopping,
where the root support diagnoses a superfluid--Mott transition
~\cite{Zheng2024}.  To expose the limited mathematical connection,
let $\bar\beta=N^{-1}\sum_j\beta_j$ and
$\Delta\beta=\max_{j,\ell}|\beta_j-\beta_\ell|$.  A leading
strong-coupling expansion of the unit-filling Bethe equations of
Ref.~\cite{Zheng2024} gives, with
$z=2\beta/U$ and
$C=(2t/U)e^{-\operatorname{Re}\bar z}$,
\[
 \log|ze^{-z}|=\log C+
 O\!\left[(\Delta\beta/U)^3\right].
\]
The cusp of the exponential map occurs at $z=1$; its closing value
$C=e^{-1}$ therefore gives the same level set
$|ze^{1-z}|=1$.  This is an asymptotic geometric correspondence,
not an identification of root species or phase structure.  The
roots of Ref.~\cite{Zheng2024} are finite-energy ground-state
variables, whereas the present $z_a$ are projective descendant
coordinates at $\lambda=\infty$ and obey the exact
truncated-exponential identity
\eqref{eq:truncated-exponential}.

In terms of the large-root Baxter factor,
\begin{equation}
 Q_{\infty}(\lambda)
 =\prod_{a=1}^r(\lambda-\lambda_a)
 \ \propto\
 \phi^{-r}L_r^{(-2S-1)}(2\phi\lambda)
 \quad (\phi\to0).
 \label{eq:Q-Laguerre}
\end{equation}
The $r$ roots have different scaled coordinates $x_a$, but on the
compactified rapidity sphere they all meet at the single point
$\lambda=\infty$.  The finite spin roots converge to the
highest-weight configuration, and Eq.~\eqref{eq:charge-BAE} shows
that the charge roots do the same because every large-root factor
tends to one.  Thus the correct coalescence is at infinity in the
nested spin sector; no coalescence of finite charge rapidities is
asserted.

The $U$ independence has a simple origin.  Interaction enters the
spin equation only through the finite dressed inhomogeneities
$\xi_j=u(p_j)$.  They first contribute beyond the universal leading
large-$\lambda$ term.  Finite $U$ therefore changes the center
energy, finite roots, and subleading approach to the EP, but not the
number $r$ of descendant roots at infinity or the block length
$2S+1$.

Figure~\ref{fig:root-contraction} provides a direct finite-volume
check using the \emph{coupled} charge and spin equations
\eqref{eq:charge-BAE}--\eqref{eq:spin-BAE}, rather than one-particle
bath-dressed poles.  The charge momenta remain finite and approach
their periodic values, while the descendant spin roots diverge and
their scaled coordinates converge to the Laguerre configuration.

\begin{figure}[p]
\centering
\includegraphics[width=0.96\textwidth]{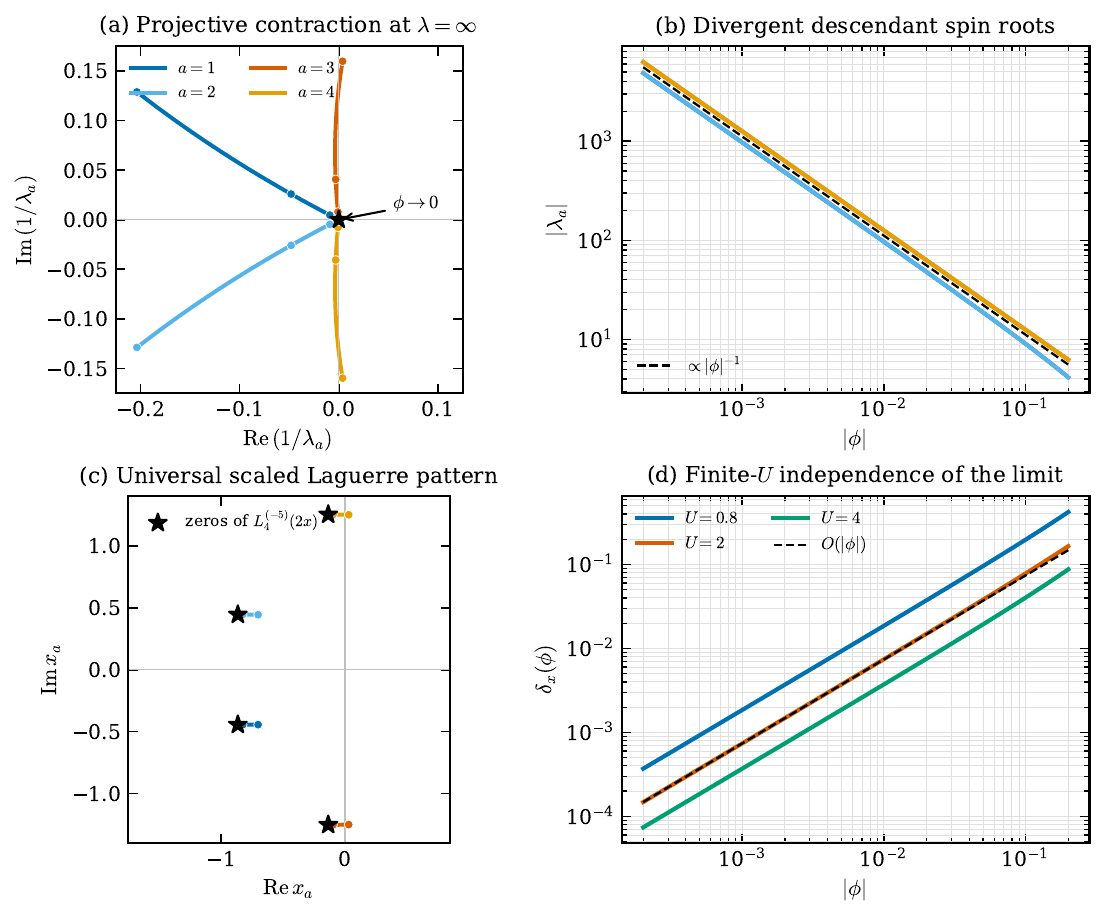}
\caption{\textbf{Exact finite-volume descendant-root contraction.}
The coupled twisted Bethe equations
\eqref{eq:charge-BAE}--\eqref{eq:spin-BAE} are solved for
$N_e=4$, $M_0=0$, $S=2$, descendant level $r=4$, $L=20$, and
$\Gamma_A=0.25$, with charge quantum numbers
$I_j=1,2,3,4$ and $\epsilon_d=-U/2$.
\textbf{(a)} For $U=2$, the inverse spin rapidities
$1/\lambda_a$ meet at the origin as $\phi\to0$, displaying the
single projective coalescence point $\lambda=\infty$.
\textbf{(b)} The four roots obey
$|\lambda_a|\propto|\phi|^{-1}$; conjugate partners overlap in
magnitude.
\textbf{(c)} The scaled coordinates $x_a=\phi\lambda_a$ converge
to the zeros of $L_4^{(-5)}(2x)$ (black stars).
\textbf{(d)} The permutation-matched error
$\delta_x(\phi)=\min_{\pi}\max_a
|x_a-x_{\pi(a)}^{\rm Lag}|$ is $O(|\phi|)$ for
$U=0.8,2,4$, demonstrating that finite $U$ changes only the
subleading approach.  The maximum unscaled Bethe-equation residual
over all displayed points is $1.18\times10^{-14}$.  No
coalescence of finite Anderson charge roots is asserted.}
\label{fig:root-contraction}
\end{figure}
\clearpage

\begin{figure}[t]
\centering
\includegraphics[width=0.96\textwidth]{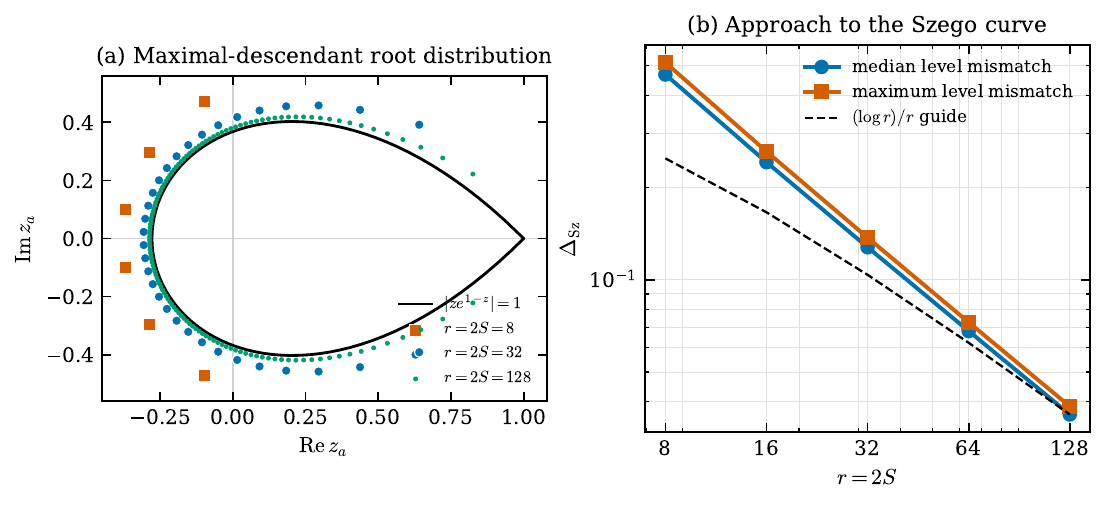}
\caption{\textbf{Large-multiplet descendant-root distribution.}
The maximal-descendant sequence $r=2S$ is evaluated after the
fixed-$r$ untwisting limit, using
$z_a=2x_a/r$.  Equation~\eqref{eq:truncated-exponential} makes these
points the exact zeros of the rescaled truncated exponential
$p_r(rz)$.
\textbf{(a)} Roots for $r=8,32,128$ approach the Szeg\H{o} contour
$\Gamma_{\rm Sz}$ (black line); only three root sets are displayed
to keep the complex plane legible.
\textbf{(b)} The median and maximum level mismatches
$\Delta_{\rm Sz}=||z_a e^{1-z_a}|-1|$ decrease monotonically for
$r=8,16,32,64,128$.  The dashed $(\log r)/r$ line is a visual guide
and is not used in the proof of Corollary~\ref{cor:Szego}.
All plotted roots were obtained by high-precision simultaneous
root finding; the largest relative polynomial residual is below
$6\times10^{-82}$.  This figure concerns the universal scaled
descendant coordinates, not a finite-$U$ phase diagram or a
coalescence of finite charge roots.}
\label{fig:Szego}
\end{figure}

\subsection{From Bethe descendants to generalized eigenvectors}

For a regular diagonal twist, the members of a pseudospin multiplet
are separated and possess ordinary Bethe vectors.  Removing the twist
sends the descendant roots to infinity; the asymptotic creation
operator becomes a global lowering generator.  The singular
similarity ${\cal C}_\phi$ in
Eq.~\eqref{eq:transfer-similarity} then maps these independent
descendants into vectors whose leading directions coincide.
Confluent, or divided-difference, combinations remove the common
leading direction and yield the generalized vectors satisfying
Eq.~\eqref{eq:chain-relation}.  This is the Bethe realization of the
Jordan-chain construction, rather than an additional assumption
about finite charge-root coalescence.

\section{Finite-matrix validation}
\label{sec:validation}

The identities above are analytic.  We additionally checked the
rational monodromy as explicit matrices to catch tensor-ordering and
twist-convention errors.  For generic complex spectral parameter and
inhomogeneities, the normalized residuals of YBE, RLL, RTT, and
$[t_G(z),t_G(w)]$ are at double-precision roundoff.  More
diagnostically, for
\begin{equation}
 G_0=\id+0.7E_{12},\qquad
 (\xi_1,\xi_2,\xi_3,\xi_4)=(-0.60,0.15,0.83,1.27),
 \label{eq:numerical-data}
\end{equation}
with $z=0.37+0.19\ii$, the Jordan partitions in
Table~\ref{tab:Jordan-checks} are obtained.

\begin{table}[ht]
\caption{Jordan blocks of a generic unipotently twisted rational
transfer matrix.  The partitions coincide with the dimensions of
the irreducible pseudospin factors in
$(\mathbb C^2)^{\otimes N_e}$.}
\label{tab:Jordan-checks}
\begin{ruledtabular}
\begin{tabular}{ccc}
$N_e$ & pseudospin decomposition & Jordan-block sizes\\
\hline
2 & $V_1\oplus V_0$ & $3,1$\\
3 & $V_{3/2}\oplus2V_{1/2}$ & $4,2,2$\\
4 & $V_2\oplus3V_1\oplus2V_0$ & $5,3,3,3,1,1$
\end{tabular}
\end{ruledtabular}
\end{table}

For each block of size $d$, the computed nullities of
$(t-\tau\id)^k$ are $1,2,\ldots,d$.  The characteristic polynomial
agrees with that of the periodic transfer matrix, as required by
Proposition~\ref{prop:isospectral}.  These calculations validate the
generic individual-transfer-matrix statement for $N_e\le4$; they are
not used as a substitute for Theorems~\ref{thm:RTT},
\ref{thm:manybody-EP}, and \ref{thm:root-contraction}.

\section{Implications and limits}
\label{sec:scope}

The construction has six direct implications.

\begin{enumerate}
\item \emph{Finite-$U$ integrability survives the EP.}
The interacting bulk matrix is the standard Anderson $R$-matrix in
the dressed rapidity.  The EP is a non-semisimple boundary twist, not
an additive deformation of $R$.

\item \emph{The many-body EP is enhanced, not generically lifted, by
the standard interaction.}
At fixed $S$, finite $U$ shifts $E_{\alpha S}$ and reorganizes the
orbital Bethe data, while exact global $GL(2)$ covariance protects a
single block of length $2S+1$.

\item \emph{Root coalescence has a precise location.}
The finite Anderson charge roots need not merge.  The universal
coalescence occurs at $\lambda=\infty$ in the nested spin problem,
and its scaled pattern is the Laguerre configuration
\eqref{eq:Laguerre}.

\item \emph{The maximal descendant has a continuum root measure.}
In the sequential untwisting and large-multiplet limits,
$z_a=2x_a/r$ becomes distributed on the Szeg\H{o} contour according
to Eq.~\eqref{eq:Szego-measure}.  This is a distribution of
projective symmetry descendants, not a new finite-energy phase.

\item \emph{The order can grow with particle number.}
A local rank-one nilpotent matrix acts through the many-body
representation $J^+_S$.  The possible order therefore reaches
$N_e+1$ in a maximally polarized multiplet, even though the local
one-particle matrix is only $2\times2$.

\item \emph{The twist has a four-state graded lift.}
The exterior-algebra matrix
$\widehat G=1\oplus G\oplus\det G$ is an exact symmetry of the
fermionic Hubbard $R$-matrix.  This supplementary compatibility
result preserves the graded Hubbard RTT algebra but is not used as
proof of the continuum Anderson construction: a model-identifying
four-state defect $L$-operator or boundary $K$-matrix is not claimed.
\end{enumerate}

These statements rely on the exact symmetry of
Assumption~\ref{ass:reduction}.  Curvature, unequal folded velocities,
nonscalar hybridization, component-dependent interactions, or an
energy-dependent matrix $M$ can break the global $GL(2)$ covariance
and may split or move the many-body EP.  Likewise, a microscopic
self-consistent treatment in which $U$ renormalizes the effective
parameters $\beta$ and $\gamma$ can shift the physical locus
$\beta^2=\gamma^2$; that is a different model.  Within
Eq.~\eqref{eq:H-M}, the locus and the block order are exact and
$U$ independent.

No thermodynamic Bethe ansatz or observable critical exponent is
deduced from the Jordan structure here.  The algebraic result
provides the appropriate starting point for such questions, but
non-diagonalizable transfer matrices require generalized-state
completeness and careful treatment of norms before thermodynamic or
dynamical conclusions are drawn.

\section{Conclusion}
\label{sec:conclusion}

We have established a finite-$U$ integrable realization of
many-body exceptional-point amplification in a linearized Anderson
impurity.  The spin--orbit branches supply the equal-velocity chiral
components after branchwise linearization and folding.  A constant
pseudo-Hermitian impurity matrix is then exactly a conserved
complexified pseudospin deformation of the conventional Anderson
model, or equivalently a $GL(2,\mathbb C)$ boundary twist.

The exact Anderson two-electron matrix retains its rational
difference form in the nonlinear dressed rapidity and supplies YBE,
RLL, and RTT for arbitrary particle number.  CPT and non-unitary
rotations preserve this algebra because the rational $R$-matrix is
$GL(2)$ invariant.  At the local EP the twist becomes unipotent while
remaining invertible.

The same twist also admits the canonical four-state lift
$1\oplus G\oplus\det G$, which the Supplementary Material proves to
be compatible with the graded $16\times16$ fermionic Hubbard
$R$-matrix.  We keep this algebraic extension separate from the
continuum proof: no four-state Hubbard defect operator reproducing
the Hamiltonian~\eqref{eq:H-M} is asserted.

The resulting spectral statement is exact: every pseudospin-$S$
multiplet is converted into one Jordan block of order $2S+1$.
The nested Bethe equations give the same contraction geometrically:
the descendant roots approach the point at infinity with scaled
positions fixed by $L_r^{(-2S-1)}$.  Finite $U$ remains fully present
in the dressed charge data and the multiplet energy but drops out of
the symmetry-controlled Jordan order.  For the maximal descendant,
the additional $r\to\infty$ scaling turns this polynomial into a
truncated exponential and produces the explicit Szeg\H{o}-curve
root measure.  This separates the robust mathematical construction
from microscopic corrections that break the equal-velocity,
scalar-coupling assumptions.

\section*{Data Availability}
The manuscript source, verification scripts, and numerical root data that
support this study are openly available in the version-controlled repository
\href{https://github.com/VMKPHYSMATH/jordan-twist-finite-u-anderson}{\texttt{VMKPHYSMATH/jordan-twist-finite-u-anderson}}.
The verified v0.1.0 snapshot is archived at
\href{https://doi.org/10.5281/zenodo.21794458}{doi:10.5281/zenodo.21794458};
the all-versions record is
\href{https://doi.org/10.5281/zenodo.21794457}{doi:10.5281/zenodo.21794457}.

\begin{acknowledgments}
The author thanks colleagues and readers whose technical questions
helped sharpen the distinction between a bulk $R$-matrix deformation,
a non-unitary change of basis, and a singular boundary twist.
\end{acknowledgments}

\end{document}